\newif\ifpublic
\publictrue %%%%  Author's notes and TODO's suppressed
\documentclass[11pt]{article}

\ifpublic
\usepackage[disable]{todonotes}
\else
\usepackage[colorinlistoftodos]{todonotes}
\fi

\usepackage{fullpage}
\usepackage{amsmath,amsfonts,amsthm,xspace,graphicx,relsize,bm}
	\usepackage[bookmarks,colorlinks,breaklinks]{hyperref}  % PDF hyperlinks, with coloured links
	\hypersetup{linkcolor=blue,citecolor=blue,filecolor=blue,urlcolor=blue} % all blue links
\usepackage{mathrsfs}
\usepackage{mathpazo}
\usepackage[margin=1in]{geometry}
\usepackage{float}
\usepackage{endnotes}
\usepackage{enumitem}
\usepackage{color}
\usepackage{amssymb,latexsym}
\usepackage{multicol}
\usepackage{pdfpages}
\usepackage[capitalize]{cleveref}
\newtheorem{theorem}{Theorem}
\newtheorem{proposition}[theorem]{Proposition}
\newtheorem{lemma}[theorem]{Lemma}
\newtheorem{claim}[theorem]{Claim}

\newtheorem{corollary}[theorem]{Corollary}
\newtheorem{definition}[theorem]{Definition}
\theoremstyle{definition}

\newtheorem{remark}[theorem]{Remark}

\newtheorem{conjecture}[theorem]{Conjecture}

\DeclareMathOperator*{\Ex}{\mathbb{E}}

\newcommand{\N}{\ensuremath{\mathbb{N}}}

\newcommand{\R}{\ensuremath{\mathbb{R}}}
\newcommand{\Z}{\ensuremath{\mathbb{Z}}}

\renewcommand{\P}{\mathsf{P}}
\newcommand{\Q}{\mathsf{Q}}

\newcommand{\eps}{\varepsilon}

\newcommand{\Sbar}{\overline{S}}

\newcommand{\xhat}{\what{x}}

\newcommand{\mt}{{-t}}

\newcommand{\X}{\mathscr{X}}
\newcommand{\Y}{\mathcal{Y}}
\newcommand{\A}{\mathscr{A}}

\newcommand{\mi}{{-i}}
\newcommand{\mP}{\mathcal{P}}

\newcommand{\wt}[1]{\widetilde{#1}}
\newcommand{\what}[1]{\widehat{#1}}

\newcommand{\val}{\mathrm{val}}

\newcommand{\Xvec}{X_{[n]}}

\newcommand{\footnoteremember}[2]{
\footnote{#2}
\newcounter{#1}
\setcounter{#1}{\value{footnote}}
}
\newcommand{\footnoterecall}[1]{
\footnotemark[\value{#1}]
}

\begin{document}

\title{Multiplayer parallel repetition for expanding games}

\author{Irit Dinur\footnoteremember{fn:funding}{Research of the first and second author supported in part by an ISF-UGC grant number 1399/14.}\thanks{Weizmann Institute of Science. email: {\tt irit.dinur@weizmann.ac.il}} \and Prahladh Harsha\footnoterecall{fn:funding} \thanks{Tata Institute of Fundamental Research. email: {\tt prahladh@tifr.res.in, rakesh09@gmail.com}} \and Rakesh Venkat\footnotemark[\value{footnote}]\and Henry Yuen\thanks{UC Berkeley. email: {\tt hyuen@cs.berkeley.edu}}}
%\date
\maketitle

\begin{abstract}
	We investigate the value of parallel repetition of one-round games with any number of players $k\ge 2$. It has been an open question whether an analogue of Raz's Parallel Repetition Theorem holds for games with more than two players, i.e., whether the value of the repeated game decays exponentially with the number of repetitions. Verbitsky has shown, via a reduction to the density Hales-Jewett theorem, that the value of the repeated game must approach zero, as the number of repetitions increases. However, the rate of decay obtained in this way is extremely slow, and it is an open question whether the true rate is exponential as is the case for all two-player games.
	
	Exponential decay bounds are known for several special cases of multi-player games, e.g., free games and anchored games. In this work, we identify a certain expansion property of the base game and show all games with this property satisfy an exponential decay parallel repetition bound. Free games and anchored games satisfy this expansion property, and thus our parallel repetition theorem reproduces all earlier exponential-decay bounds for multiplayer games. More generally, our parallel repetition bound applies to all multiplayer games that are \emph{connected} in a certain sense.
	
	We also describe a very simple game, called the GHZ game, that does \emph{not} satisfy this connectivity property, and for which we do not know an exponential decay bound. We suspect that progress on bounding the value of this the parallel repetition of the GHZ game will lead to further progress on the general question.
\end{abstract}

\section{Introduction and Results}

%\subsection{NotBackground}

We consider multi-player one-round games, and their parallel repetition. In a $k$-player game $G$, a referee chooses a $k$-tuple of questions $(x^1,\ldots,x^k)$ from some question distribution $\mu$, and sends $x^t$ to player $t$. Each player $t$ gives an answer $a^t$ that only depends on their question (i.e., they cannot communicate with each other). The referee evaluates the players' questions and answers according to some predicate $V((x^1,\ldots,x^k),(a^1,\ldots,a^k))$, and the players win if this predicate evaluates to $1$. The \emph{value} of a game $G$, denoted by $\val(G)$, is the players' maximum success probability over all possible strategies they may use.

Here is a very natural operation on games, called \emph{parallel repetition}: starting with a $k$-player game $G$, we can construct a new $k$-player game $G^{\otimes n}$, called the $n$-fold parallel repetition of $G$. In $G^{\otimes n}$, the referee will select $n$ independent question tuples $(x^1_i,\ldots,x^t_i)$ from $\mu$ for each \emph{coordinate} $i = 1,\ldots,n$, and send $(x^t_1,\ldots,x^t_n)$ to each player $t$. Each player has to respond with $n$ answers, and they win this repeated game if their answers and questions for each coordinate $i$ satisfy the original game predicate $V$. We call $G$ the \emph{base game} of the parallel repeated game $G^{\otimes n}$.

The central question we consider is how $\val(G^{\otimes n})$ depends on the base game $G$ and the number of repetitions $n$. When $G$ is a two-player game, the behavior of $\val(G^{\otimes n})$ has been extensively studied, especially due to its applications to probabilistically checkable proofs and hardness of approximation. The central result in this area is Raz's Parallel Repetition Theorem~\cite{Raz1998} (coupled with subsequent improvements due to Holenstein~\cite{Holenstein2009}), which states the following:

\begin{theorem}[Two-player parallel repetition]
Let $G$ be a one-round two-player game with $\val(G) \leq 1 - \eps$  for some $\eps \in (0,1)$. Then for all $n \geq 0$,
$$
	\val(G^{\otimes n}) \leq \exp \left (- \frac{c \eps^3 n}{\log |\A|} \right)
$$
where $|\A|$ is the answer alphabet size of $G$ and $c> 0$ is a universal constant.
\end{theorem}
\noindent In other words, for \emph{nontrivial} two-player games $G$ (i.e., games whose value is less than $1$), $\val(G^{\otimes n})$ decays exponentially fast in $n$.

What about parallel repetition for games involving more than two players? It remains an intriguing open question  whether Raz's Theorem can be extended to the multiplayer case. An early result of Verbitsky~\cite{Verbitsky1996} shows that for multiplayer games $G$ with $\val(G) < 1$, the value of the repeated game $G^{\otimes n}$ must decay to $0$ as $n$ goes to infinity. However, the bound on the rate of decay is extremely weak: his result only shows that $\val(G^{\otimes n})$
is bounded by a function that is inversely proportional to the inverse Ackermann function of $n$~\cite{Polymath2012}! This poor rate of decay comes from its black-box usage of the density Hales-Jewett theorem from extremal combinatorics. %In contrast, when $G$ is a nontrivial two-player game, Raz's theorem shows that one only needs $n = O(\log 1/\delta)$, if we treat the alphabet size and value of the game $G$ as constant.

So far, Verbitsky's theorem is still the only result that gives a general parallel repetition bound for all multiplayer games. Exponential-decay parallel repetition bounds (\`{a} la Raz) for multiplayer games have been proven when there are additional assumptions on the game; for example, it has long been a folklore result that multiplayer \emph{free} games satisfy an exponential-decay parallel repetition theorem~\cite{ChungWY2015}.\footnote{A free game is one where each players' question is independent of all the other players'.} Recently, Bavarian, Vidick and Yuen~\cite{BavarianVY2015} studied a variant of parallel repetition (called ``anchoring'') where the base game $G$ is first modified to an equivalent game $\widetilde{G}$ before being repeated in parallel, producing $\widetilde{G}^{\otimes n}$. They proved that the value of $\widetilde{G}^{\otimes n}$ is exponentially small in $n$ when $\val(G) < 1$, and otherwise $\val(\widetilde{G}^{\otimes n}) = 1$.\footnote{In fact, Bavarian, Vidick and Yuen were motivated by the question of parallel repetition for \emph{quantum} players; they showed that so-called ``anchored''  games satisfy quantum parallel repetition theorems.} Buhrman et al.~\cite{BuhrmanFS2014} show that the \emph{non-signalling} value (which upper bounds $\val(G)$) of a multiplayer game decays exponentially under parallel repetition if the game has \emph{full support}, that is, all possible question tuples to the players occur with non-zero probability. However, this does not necessarily imply a non-trivial decay of the value in general for such games, since there are games $G$ with non-signalling value $1$, but $\val(G)<1$.

We observe that the class of multiplayer games for which we have exponential-decay parallel repetition bounds (or, for that matter, any rate of decay better than inverse Ackermann!) all share a particular feature in common: when viewed as hypergraphs, the games all possess a certain ``well-connectedness'' property. For example, consider any two question tuples $x,\what{x}$ in the support of the question distribution $\mu$ of a free game. The question tuple $x = (x^1,\ldots,x^k)$ can be ``locally morphed'' to $\what{x} = (\what{x}^1,\ldots,\what{x}^k)$ via a sequence of question tuples $(\what{x}^1,\ldots,\what{x}^j,x^{j+1},\ldots,x^k)$ for $j = 1,\ldots,k$, each of which remain in the support of $\mu$. Furthermore, the anchoring transformations of~\cite{BavarianVY2015} can be understood as improving the connectivity properties of the base game before repetition. In this paper, we formalize this well-connectedness property as a type of \emph{expansion} of the base game, and show that any connected multiplayer game has exponential-decay parallel repetition bounds. We associate with every base game $G$, a related graph $H_G$ (see Definition~\ref{def:oneconnected}) and show that if $H_G$ is connected, then the value of the repeated game, $\val(G^{\otimes n})$, goes down exponentially in $n$, more precisely, for sufficiently large $n$, $$\val(G^{\otimes n}) \leq \exp\left(-\frac{c \eps^5 \lambda^2 n}{\log |\A|}\right),$$ where $\lambda$ is the spectral gap of the Laplacian of the graph $H_G$ and $c$ is some universal constant (see Theorem~\ref{thm:expanding-main} for an exact statement of the result). Thus, if the graph $H_G$ is connected (i.e., $\lambda > 0$), we have an exponential decay in $n$. In the case of games $G$, wherein the associated graph $H_G$ is not only connected but also expanding (i.e., $\lambda$ is a constant), as is the case with free games, and anchored games, the rate of exponential decay is a function of alphabet size $|\A|$ of the base game $G$ as in Raz's theorem.  The class of \emph{full-support} games investigated in Buhrman et al.~\cite{BuhrmanFS2014} have a connected $H_G$, and hence our result shows an exponential decay bound on the value of such games under parallel repetition.

\paragraph*{Why care about games with more than two players?} The notion of a game is an extremely basic notion, and it's use is pervasive in communication complexity, probabilistically checkable proofs (PCPs), etc. Whereas two-player games are already quite powerful and give us a lot, many problems are inherently higher-dimensional, i.e., would more naturally be cast as games with more than two players. The reason this is not commonly done is because we don't know how to analyze these creatures.
For example, constraint-satisfaction-problems with arity $k$ are naturally cast as a $k$-player game. They can be reduced to a two-player game in the same way that a hypergraph can be converted to a graph, but this reduction in dimensionality might be lossy.
Indeed, it is empirically true that PCPs with 3 or more queries are much more powerful than 2-query PCPs, but what is the reason for this?

Furthermore, this sudden jump in difficulty in going from two-player problems to three or more players is encountered also when studying multiparty communication complexity, and seemingly because of the same technique limitations. While direct sum and direct product theorems are known for two-party communication complexity, nothing is known for the multiparty setting (in the so-called \emph{number-on-forehead} model), and in fact making progress on this is connected to hard problems in circuit complexity.

We feel that the study of games with three or more players is a very important component in understanding such questions.

\subsection{Notation}

We establish some notational conventions, before stating our results formally.

For a $k$-player game $G$, we will let $\X^t$ denote the question alphabet for player $t$, and $\X = \X^1 \times  \X^2 \times \cdots \times \X^k$ is the question alphabet for all the players together, underlying the question distribution $\mu$. We will let $\A^t$ denote the answer alphabet for player $t$, and $\A = \A^1 \times \cdots \times \A^k$ to denote the answer alphabet for all the players together.

 We will use superscripts to denote the players, and subscripts to denote the coordinate in parallel repetition. For example, $x^t_i$ denotes the question received by player $t$ in coordinate $i$.
 A single variable $x$ can denote questions to the players in some coordinate clear from context, or a single coordinate game. We use $x^{\mt}$ to denote the questions to all but the $t$-th player in a single coordinate. When talking about multiple coordinates, we will use subscripts: $x_{\mi}$ denotes the questions to players in all but the $i$'th coordinate, and $x_\mi^t$ denotes the all questions to player $t$ in the repeated game except for the $i$'th coordinate. To denote the question to player $t$ in all coordinates, we use $x^t_{[n]}$.

We largely adopt the notational conventions from~\cite{Holenstein2009} for probability distributions. We let capital letters denote random variables and lower case letters denote specific samples. We use $\P_X$ to denote the probability distribution of random variable $X$, and $\P_X(x)$ to denote the probability that $X = x$ for some value $x$. For multiple random variables, e.g., $X, Y, Z$, $\P_{XYZ}(x,y,z)$  denotes their joint distribution with respect to some probability space understood from context.

We use $\P_{Y | X = x}(y)$ to denote the conditional distribution $\P_{YX}(y,x)/\P_X(x)$, which is defined when $\P_X(x) > 0$. When conditioning on many variables, we usually use the shorthand $\P_{X | y,z}$ to denote the distribution $\P_{X | Y =y,Z=z}$. For an event $W$ we let $\P_{X Y | W}$ denote the distribution conditioned on $W$. We use the notation $\Ex_{X} f(x)$ and $\Ex_{\P_X} f(x)$ to denote the expectation $\sum_{x} \P_X(x) f(x)$.

Let $\P_{X_0}$ be a distribution over $\X$ and $\P_{X_1,Y}$ a joint distribution over $\X \times \Y$. Suppose for every $x$ in the support of $\P_{X_0}$, the conditional distribution $\P_{Y | X_1 = x}$ defined over $\Y$ is well-defined. We then define the distribution $\P_{X_0} \P_{Y | X_1}$ over $\X \times \Y$ as
$$
	(\P_{X_0} \P_{Y | X_1})(x,y) \,:=\, \P_{X_0}(x) \cdot \P_{Y | X_1 = x}(y).
$$
%Additionally, we write $\P_{X_0 Z} \P_{Y | X_1}$ to denote the distribution $(\P_{X_0 Z} \P_{Y | X_1})(x,z,y) := \P_{X_0 Z}(x,z) \cdot \P_{Y | X_1 = x}(y)$.

For two random variables $X_0$ and $X_1$ over the same set $\X$, we use
$$\| \P_{X_0} - \P_{X_1} \| \,:=\, \frac{1}{2}\sum_{x \in \X} |\P_{X_0}(x) - \P_{X_1} (x)|,$$
to denote the total variation distance between $\P_{X_0}$ and $\P_{X_1}$.

\subsection{Our results}

To define our class of connected and expanding games, we need the following notion of the \emph{$(k-1)$-connection graph} of a game $G$. This graph, denoted $H_G$, has a vertex for every $k$-tuple of questions, and two $k$-tuples are connected by an edge if they agree on $(k-1)$ coordinates. A game is $(k-1)$-connected iff $H_G$ is connected.

To further define our notion of expansion for a $k$-player game we need to take the weights of $G$ into account when defining $H_G$.  For this it is instructive to think of an intermediate bipartite graph $B_G = (\X',\X,E)$ as follows. The right hand vertices is simply $\X$, the set of all $k$-tuples of questions, and we endow these vertices with weights as given by $G$. The left hand vertices consists of all punctured $k$-tuples, which are $k$-tuples of questions where exactly one of the entries is replaced by a special $\star$ symbol. Connect each $k$-tuple of questions to all of the $k$ ways to make it into a punctured $k$-tuple. Now, consider the distribution on punctured tuples obtained by selecting a random $k$-tuple from $\X$ according to the game distribution, and then puncturing it in a random location. The graph $H_G$ is defined by selecting a random punctured tuple according to this distribution, and then selecting independently two $k$-tuples conditioned on this puncturing. Note that each completion is distributed exactly according to the original game distribution.

We now move to a completely explicit description consistent with the above. In what follows, $\P_X(x)$ denotes the probability of question tuple $x$ under the question distribution $\mu$, $\P_{X^t}(x^t)$ denotes the marginal probability of player $t$'s question, and $P_{X^t | X^{-t} = x^{-t}}(x^t)$ denotes the same probability, conditioned on the other players having received $x^{-t}$. % (See \cref{sec:prelim} for details regarding our notation for probability distributions.)

\begin{definition}[$(k-1)$-connection graph of $G$]\label{def:oneconnected}
Let $G = (\mu,V)$ be a $k$-player game with question set $\X = \X^1 \times \cdots \X^k$. The \textbf{$(k-1)$-connection graph of $G$} is the weighted graph $H_G = (V_H, \rho)$ with vertex set $V_H = \X$ and weight function $\rho: \X \times \X \to [0,1]$, defined as follows: for every $x,x' \in \X$,
\begin{align*}
	\rho(x, x') = \left \{ \begin{array}{ll}
		 \frac{1}{k}~ \P_{X}(x) \left [ \sum_{t\in[k]} \P_{X^t | x^{\mt}}(x^t)\right ] 	& \mbox{ if } x = x',\\
	%	 \frac{1}{k}~ \P_{X^{\mt}}(x^{\mt}) \cdot\P_{X^t|X^{\mt}} (x^{t})\cdot \P_{X^t|X^{\mt}} (x'^t)  &  \mbox{ if } x^{\mt} = x'^{\mt} \mbox{ and } x^t \neq x'^t \\
	\frac{1}{k}~ \P_{X^\mt}(x^\mt) \cdot \P_{X^t|x^{\mt}} (x^{t})\cdot \P_{X^t|x^{\mt}} (x'^t) &  \mbox{ if $\exists \, t$ s.t. } x^{\mt} = x'^{\mt} \mbox{, } x^t \neq x'^t ,\\
		0 & \mbox{ otherwise. }
	\end{array}
\right.
\end{align*}
\end{definition}
\noindent The weight function $\rho(x,x')$ can be viewed as the probability of generating the pair $(x,x')$ according to  the following random process: first, $x \in \X$ is sampled from the distribution $\P_X$. Then, a coordinate $t \in [k]$ is chosen uniformly at random, and $x'$ is sampled from the conditional distribution $\P_{X | x^\mt}$ (that is, the distribution $\mu$ conditioned on $x^\mt$).

Observe that $\rho$ is symmetric, i.e., $\rho(x,x') = \rho(x',x)$. Furthermore, note that the weight on any given vertex is exactly:
\begin{align*}
	\rho(x,\cdot) &= \sum_{x'} \rho(x,x') = \\
	&= \P_{X}(x) \cdot \frac{1}{k}\sum_{t \in [k]} \P_{X^t |x^{\mt}} (x^t) +\P_{X}(x)\cdot \frac{1}{k}\sum_{t\in[k]} \sum_{x'^t \neq x^t}\P_{X^t |x^\mt} (x'^t)\\
	&= \P_{X}(x).
\end{align*}
Therefore $\rho(\cdot,\cdot)$ is a probability distribution over $\X \times \X$.

\begin{remark}
Henceforth, when we talk about graph properties such as diameter, connectedness or expansion of $H_G$, we will do so only with respect to the vertices having non-zero weight.
\end{remark}

\medskip

\noindent We now recall the definition of a graph with a weight function $\rho$ being a spectral expander:

\begin{definition}[Normalized Laplacian]
 Let $H$ be a weighted graph where $\rho(u,v) \leq 1$ is the weight between vertices $u$ and $v$.  The normalized Laplacian $L_H\in \R ^{|V| \times |V|}$ of $H$ is defined to be
\begin{align*}
	(L_H)_{u,v} = \left \{ \begin{array}{ll}
		1 - \frac{\rho(u,v)}{\rho(v)} 	& \mbox{ if } u = v \mbox { and } \rho(v) \neq 0 \\
		- \frac{\rho(u,v)}{\sqrt{\rho(u)\rho(v)}} & \mbox { if } \rho(u),\rho(v) \neq 0 \\
		0 & \mbox{ otherwise }
	\end{array}
	\right.
\end{align*}
where $\rho(u) = \sum_v \rho(u,v)$ and $\rho(v) = \sum_u \rho(u,v)$.
\end{definition}
It is well-known that the second smallest eigenvalue of $H$ is given by the following variational formula: for all $r \in \N$,
\begin{align}
	\lambda(H) = \inf_{g} \frac{\sum_{u,v} \rho(u,v) \| g(u) - g(v) \|^2}{\sum_{u} \rho(u) \| g(u) - \overline{g} \|^2}
	\label{eq:expander}
\end{align}
where the infimum is over all vector-valued functions $g: V(H) \to \R^{r}$ defined on the vertices of $H$, and $\overline{g}$ is a vector in $\R^r$ where for each $i \in [r]$, $\overline{g}_i = \sum_u \rho(u) g(u)_i$.

\begin{definition}[Expander graph]
Let $\lambda \in (0,1)$. A graph $H$ is a $\lambda$-expander if $\lambda(H) \geq \lambda$.
\end{definition}

%\begin{remark}
%\label{rm:connected}
%All connected, unweighted graphs $H$ with $|V|$ vertices are a $1/|V|^2$-expander.
%\end{remark}

Our main result is an exponential-decay parallel repetition bound for multiplayer games whose $(k-1)$-connection graph is expanding:

\begin{theorem}[Main theorem] \label{thm:expanding-main}
Let $\epsilon, \lambda \in (0,1)$. Let $G$ be a $k$-player game  with $\val(G) \leq 1-\eps$. If the $(k-1)$-connection graph $H_G$ is a $\lambda$-expander, then we have, for all $n \geq \frac{\log 4/\eps}{\eps^5 \lambda^2 }$:
$$
\val(G^{\otimes n}) \leq \exp\left(-\frac{c \eps^5 \lambda^2 n}{\log |\A|}\right)
$$
where $\A = \A^1 \times \cdots \times \A^k$ is the answer alphabet in $G$, and $c$ is a universal constant.
\end{theorem}

By applying our main theorem to free games and anchored games, we recover existing exponential-decay parallel repetition results for multiplayer games~\cite{ChungWY2015,BavarianVY2015}. We also get an exponential-decay lower bound for \emph{connected games} -- games whose $(k-1)$-connection graph is connected. We record these consequences in the following corollary:
%We obtain the following immediate corollaries, which recover existing exponential-decay parallel repetition theorems for multiplayer games:

\begin{corollary}\label{cor}
	Let $G$ be a $k$-player game with $\val(G) \leq 1 - \eps$, and let $n \geq \frac{\log 4/\eps}{\eps^5 \lambda^2 }$. If $G$ is:
	\begin{enumerate}
		\item Free, i.e., $\mu(x) = \mu^1(x^1) \times \cdots \times \mu^k(x^k)$, then
		$$
	\val(G^{\otimes n}) \leq \exp \left(-\frac{c \, \eps^5 \, n}{k^2 \log |\A|}\right).
		$$
		\item $\alpha$-Anchored (see Definition~\ref{def:anchored-game}, and \cite{BavarianVY2015}  for more details), then
		$$ \val(G^{\otimes n}) \leq \exp \left(-\frac{c \, \alpha^{2k} \,\eps^5 \,  n}{64 ~k^2\log |\A|}\right). $$
		\item Connected, i.e., the $(k-1)$-connection graph is connected, then
		$$
		\val(G^{\otimes n}) \leq \exp \left(-\frac{c \,\rho_{min}^2 \,  \eps^5 \,n}{\log |\A|}\right)
		$$
		where $ \rho_{min} = \min_{u,v : \rho(u,v) > 0} \rho(u,v)$. In particular, if the game $G$ is such that $\mu$ is the uniform distribution over some set $S \subseteq \X$, then $\rho_{min} \geq (k|S|^2 )^{-1}$.
%		\item Fortified (see~\cite{BVY16fort} for a precise definition of fortified multiplayer games);
	\end{enumerate}
	where $c$ is a universal constant.
\end{corollary}
\noindent The proof of Corollary~\ref{cor} can be found in Appendix~\ref{sec:cor}.

Observe that our proof of exponential decay for games whose corresponding $(k-1)$-connection graph is connected
proves a rate of exponential decay that is dependent on the size of the the base game $G$. It is conceivable that this rate of decay can be further improved to depend only on the alphabet size $|\A|$ of the base game and be independent of the size of the base game (as is the case in Raz's theorem for 2 player games). For games whose corresponding $(k-1)$-connection graph is expanding (as is the case with free games and anchoring games), we obtain a rate of exponential decay which is a function of only the base game's alphabet size.

\begin{remark}
For simplicity, we state Theorem~\ref{thm:expanding-main} assuming the base game has a connected $(k-1)$-connection graph. It is easy to check (from the proof of Theorem~\ref{thm:expanding-main}) that it also extends to games that are disjoint union of games each of which has a connected $(k-1)$-connection graph. By disjoint union we mean that each question occurs only in one of the components. For $k=2$, this captures all possible games since every game is a union of disjoint games whose $(k-1=2-1=1)$-connection graphs are connected). We note that for 2-player games, there are alternate proof techniques~\cite{Raz1998,Holenstein2009} using correlated sampling which prove even better rate of exponential decay (our proof of Theorem~\ref{thm:expanding-main} does not use correlated sampling).
In contrast, for $k>2$, there are many games that are not captured by our theorem. We will see below an example of such a $k=3$-player game called the GHZ game.
\end{remark}

%More generally, we prove that as long as the $1$-connection graph $H_G$ is \emph{connected}, then the game $G$ satisfies an exponential-decay parallel repetition bound:
%\begin{corollary}\label{cor:connected-game}
%Let $G$ be a $k$-player game  with $\val(G) \leq 1-\eps$. If the $1$-connection graph $H_G$ is connected (i.e. has no isolated vertices) and has weight function $\rho$, then:
%$$\val(G^{\otimes n}) \leq \exp \left(-\frac{c'' \rho_{min}^2 ~  \eps^5 n}{\log |\A|}\right)$$
%where $c''$ is a universal constant, and
%$$
%	\rho_{min} = \min_{u,v : \rho(u,v) > 0} \rho(u,v).
%$$
%
%\end{corollary}

%\begin{proof}
%This follows from the observation that $\lambda(H) \geq \rho_{min}$ when the graph $H$ is connected.
%\end{proof}

\paragraph*{A comment about fortified games.} Bavarian, Vidick and Yuen also proved a parallel repetition bound for a special class of multiplayer games \emph{fortified games}~\cite{BavarianVY2016} (a class of games introduced by Moshkovitz~\cite{Moshkovitz2014}). However, we do not consider this a ``true'' exponential-decay parallel repetition bound, because it does not establish a decay bound of the form $\val(G^{\otimes n}) \leq \exp(-\beta n)$ for some constant $\beta$ that depends on the game $G$, but is independent of $n$. Instead, it proves a decay bound that is exponential only for a small number of repetitions (depending on the base game). After this small number of repetitions, there are no guarantees about any further value decay (other than that promised by Verbitsky's theorem). Because we are interested in the asymptotic behavior of an $n$-repeated multiplayer game as $n$ goes to infinity, we do not consider the parallel repetition of fortified games here.

\paragraph*{A disconnected three-player game.} It may seem that, given Corollary~\ref{cor}, we have established a general exponential-decay parallel repetition bound for \emph{all} multiplayer games, albeit with some slightly annoying dependency on a quantity related to the minimum probability of any question from $\mu$. Unfortunately, this is far from the case.

Here is a simple three-player game called the \emph{GHZ game} whose parallel repetition resists analysis; the best decay bound we have comes from Verbitsky's theorem~\cite{Verbitsky1996}. The GHZ game is a three-player game\footnote{The GHZ game comes from the study of non-locality in quantum physics; when the players use classical strategies, their maximum success probability is $\val(G) = 3/4$, but using quantum entanglement, the GHZ can be won with certainty~\cite{GreenbergerHSZ1990}.} where the referee samples a question triple $(x,y,z)$ uniformly at random from $\{ (1,0,0), (0,1,0), (0,0,1), (1,1,1) \}$, and sends each bit of the triple to the corresponding player. The players respond with bits $a,b,c$ respectively, and they win iff $x \wedge y \wedge z = a \oplus b \oplus c$. It is easy to see that $\val(GHZ) = 3/4$ (achieved by the strategy where all players always output ``$0$''). However, the best general bound we have on $\val(GHZ^{\otimes n})$ is the weak inverse-Ackermann decay given by Verbitsky's theorem.

Our main theorem does not apply because the $(k-1)$-connection graph $H_{GHZ}$ of the GHZ game is actually \emph{disconnected}; no two question triples are connected via a single coordinate change. One necessary criterion for the $(k-1)$-connection graph to be connected is that, after fixing any subset of $(k-1)$ players' questions, the remaining player's question is yet undetermined. On the other hand, the players' questions in the GHZ game satisfy a linear relation (i.e. $x \oplus y \oplus z = 1$), and thus fixing two players' questions also fixes the third.

We believe that the strong correlations present in the GHZ question distribution represent the ``hardest instance'' of the multiplayer parallel repetition problem. Existing techniques from the two-player case (which we leverage in this paper) appear to be incapable of analyzing games with question distributions with such strong correlations. Thus we explicitly raise the open question of proving an exponential-decay parallel repetition bound for the GHZ game:

\begin{conjecture}[GHZ parallel repetition]
	There exists a constant $\beta > 0$ such that for all $n$, $$\val(GHZ^{\otimes n}) \leq \exp(-\beta n).$$
\end{conjecture}

Finally, we remark that this challenge of handling strongly correlated question distributions is reminiscent of the challenge of  proving \emph{direct sum} theorems for multiparty communication complexity in the \emph{Number-on-Forehead} (NOF) model. There, each player sees every players' inputs but their own, so fixing $(k-1)$ out of $k$ players' inputs will fix the remaining player's inputs. Proving direct sum results in NOF communication complexity has resisted progress for reasons that appear to be related to the multiplayer parallel repetition problem.

%Although the GHZ game `feels' connected : the marginal on any pair of players questions is, in fact, a free game, its $1$-connected graph consists of isolated vertices. Our theorems, thus, do not yield anything for this game.

\section{Proof of Theorem~\ref{thm:expanding-main}}

\subsection{Proof outline}

We first give a brief overview of the information-theoretic approach to proving two-player parallel repetition as in~\cite{Raz1998,Holenstein2009}, and explain the technical barrier to extending the proof to three or players. We then will describe how we circumvent this technical barrier.

Essentially all known proofs of parallel repetition proceed via reduction, showing how a ``too good'' strategy  for the repeated game $G^n$ can be ``rounded'' into a strategy for $G$ with success probability strictly greater than $\val(G)$, yielding a contradiction. 

Let $\mathscr{S}^n$ be a strategy for $G^n$ that has a high success probability. Either by induction or via a probabilistic argument one can identify a set of coordinates $S$ and an index $i$ such that \\ $\Pr(\text{Players win round $i$} | W_S) > \val(G) + \delta$, where $W$ is the event that the players' answers satisfy the predicate $V$ in all instances of $G$ indexed by $S$. Given a pair of questions $(x,y)$ in $G$ the strategy $\mathscr{S}$ embeds them in the $i$-th coordinate of a $n$-tuple of questions
$$x_{[n]} y_{[n]} = \binom{x_1, x_2, \ldots, x_{i-1}, \qquad x \qquad , x_{i+1}, \ldots, x_n}{y_1, y_2, \ldots, y_{i-1}, \qquad y \qquad , y_{i+1}, \ldots, y_n}$$
that is distributed according to $\P_{X_{[n]} Y_{[n]} | X_i = x, Y_i = y,W}$. The players then simulate $\mathscr{S}^n$ on  $x_{[n]}$ and $y_{[n]}$ respectively to obtain answers $(a_1,\ldots,a_n)$ and $(b_1,\ldots,b_n)$, and return $(a_i,b_i)$ as their answers in $G$. The single-shot strategy $\mathscr{S}$ succeeds with probability precisely $\Pr(\text{Win $i$} | W_S)$ in $G$, yielding the desired contradiction.

As $\mathscr{S}^n$ need not be a product strategy, conditioning on $W_S$ may introduce correlations that make $\P_{X_{[n]} Y_{[n]} | X_i = x, Y_i = y,W_S}$ impossible to sample exactly. A key insight in Raz' proof of parallel repetition is that it is still possible for the players to \emph{approximately} sample from this distribution.
For this, we introduce a \emph{dependency-breaking variable} $R$ with the following properties:
\begin{enumerate}
\item[(a)] Given $r\sim \P_R$ the players can locally sample $x_{[n]}$ and $y_{[n]}$ according to \\ $\P_{X_{[n]} Y_{[n]} | X_i = x, Y_i = y, W_S}$,
\item[(b)] The players can jointly sample from $\P_R$ using shared randomness.
\end{enumerate}

In~\cite{Holenstein2009} $R$ is defined so that a sample $r$ fixes at least one of $\{x_{i'}, y_{i'}\}$ for each $i' \neq i$. It can then be shown that conditioned on $x$, $R$ is nearly (though not exactly) independent of $y$, and vice-versa. In other words,

\begin{equation}
\label{eq:intro_cor_samp}
	\P_{R | X_i = x, W_S} \approx \P_{R | X_i = x, Y_i = y, W_S} \approx \P_{R | Y_i = y, W_S}
\end{equation}
where ``$\approx$'' denotes closeness in statistical distance. Eq.~\eqref{eq:intro_cor_samp}  suffices to guarantee that the players can \emph{approximately} sample the same $r$ from $\P_{R | X_i = x, Y_i = y, W_S}$ with high probability, achieving point (b) above. This sampling is accomplished through a technique called \emph{correlated sampling}.

This argument relies heavily on the assumption that there are only two players who employ a deterministic strategy. With more than two players, it is not known how to design an appropriate dependency-breaking variable $R$ that satisfies \emph{both} items (a) and (b) above: in order to be jointly sampleable, $R$ needs to fix as few inputs as possible; in particular, no single player should require knowledge of the other player's questions to sample $R$. On the other hand, in order to allow players to locally sample their inputs conditioned on $R$, the variable needs to fix as many inputs as possible. These two requirements turn out to be in direct conflict as soon as there are more than two players, and a straightforward generalization of the two-player version of the dependency-breaking variable cannot be ``correlatedly sampled'' by all players, unless every player has  knowledge of the question received by some other player.

%Our proof starts out with arguments seen in previous proofs: if there is a great strategy for $G^{\otimes n}$, we can extract out an impossibly good one for $G$ with value $\geq 1-\eps$, leading to a contradiction. However, in order to actually extract out a strategy for the single coordinate game, the players had to go through \emph{correlated sampling} to sample some questions and answers in a subset of the $n$ coordinates (altogether called a \emph{dependency-breaking variable}) and then simulate $G^{\otimes n}$ on these sampled questions. In the two-player version,  each player had enough information in hand to go about doing this (i.e.  they only need knowledge of their own question). The main difficulty in extending this to three or more players is that the straightforward generalization the dependency-breaking variable cannot be ``correlatedly sampled'' without knowledge of some other player's question.
%the distribution each  players needs to carry out correlated sampling requires knowledge of some other player's question too.
\medskip

We avoid this roadblock by proving in Section~\ref{sec:avoid-correlated} that if the $(k-1)$-connected graph $H_G$ is connected, then the players can avoid correlated sampling altogether. In fact, they can sample an appropriate dependency-breaking variable from a \emph{global} distribution that does not depend on any player's question.

%\smallskip
%By side-stepping the correlated sampling step, our final theorem statement avoids a direct dependence on $k$ (though it depends on the combined answer set size of the $k$ players). However, the use of the expansion property incurs a worse dependence on $\eps$ than in the two-player version, where $\eps = 1 - \val(G)$.

\subsection{Following Raz-Holenstein}

Fix a $k$-player game $G = (\mu,V)$, with answer alphabet $\A = \A_1 \times \cdots \times \A_k$ and $\val(G) = 1 - \eps$. Consider the $n$-fold parallel repetition $G^{\otimes n}$ and consider an optimal strategy $\{ f^t : (\X^t)^{\otimes n} \to (\A^t)^{\otimes n}\}_{t\in [k]}$ for the $k$ players.

\smallskip
For $i \in [n]$, let $W_i$ denote the event that the players win coordinate $i$ using this optimal strategy. Let $W = W_1 \wedge \cdots \wedge W_n$ denote the event that the players win all coordinates. For a set $S \subseteq [n]$, let $W_S= \wedge_{i \in S} W_i$. In the following, all probabilities are with respect to this optimal strategy.

\medskip

\begin{proposition}
\label{prop:subset}
	Let $\eps > 0$. Suppose that $\log 1/\Pr(W) \leq \eps n/16 - \log 4/\eps$. Then there exists a set $S \subseteq [n]$ of size at most $t = \frac{8}{\eps} \left( \log 4/\eps + \log 1/\P(W) \right)$ such that
	$$
		\P_{i \notin S} (\neg W_i | W_S) \leq \eps/2
	$$
	where $i$ is chosen uniformly from $[n] - S$.
\end{proposition}
\begin{proof}
	Set $\delta = \eps/8$. Let $W_{> 1 - \delta}$ denote the event that the players won more than $(1 - \delta)n$ rounds. To show existence of such a set $S$, we will show that $\Ex_S \P(\neg W_i | W_S) \leq \eps/2$, where $S$ is a (multi)set of $t$ independently chosen indices in $[n]$. This implies that there exists a particular set $S$ such that $\P(\neg W_i | W_S) \leq \eps/2$, which concludes the claim.
	
	First we write, for a fixed $S$,
	\begin{align*}
		\P ( \neg W_i | W_S) &= \Pr(\neg W_i | W_S, W_{> 1 - \delta}) \P(W_{> 1 - \delta} | W_S)  \\
		&\qquad \qquad + \P(\neg W_i | W_S, \neg W_{> 1 - \delta}) \P(\neg W_{> 1 - \delta} | W_S).
	\end{align*}
	Observe that $\P(\neg W_i | W_S \wedge W_{> 1 - \delta})$ is the probability that, conditioned on winning all rounds in $S$, the randomly selected coordinate $i \in [n] - S$ happens to be one of the (at most) $\delta n$ lost rounds. This is at most $\delta n/(n - t) \leq \eps/4$. Now observe that
	\begin{align*}
		\Ex_S \P(\neg W_{> 1 - \delta} | W_S) &\leq \Ex_S \frac{\P(W_S| \neg W_{> 1 - \delta})}{\P(W_S)} \leq \frac{1}{\P(W)} (1 - \delta)^t \leq \eps/4
	\end{align*}
	where for the second inequality we used the fact that $\P(W_S) \geq \P(W)$.	
\end{proof}

For the remainder of this proof we will fix a set $S$ as given by Proposition~\ref{prop:subset}. By renaming coordinates, we will assume without loss of generality that $S$ is the last $t$ coordinates of $[n]$. We will let $m = n - |S|$. We will refer to the games indexed by set $S$ as the $S$-games.

\subsection{Dependency-breaking variables}

We define the $k$-player analogue of the dependency-breaking variable $R$ that is used so crucially in information-theoretic proofs of parallel repetition~\cite{Raz1998,Holenstein2009,BravermanG2015}. $R$ will consist of a variable $\Omega$, which fixes the questions for the $S$-games, and at least $(k-1)$-of-$k$ questions in every other coordinate, and a variable $Z = (A_S)$, which fixes the answers of $S$-games. More formally, $\Omega = (\Omega_1,\ldots,\Omega_m,X_S)$, where $X_S$ are fixed questions for the $S$-games. Each $\Omega_i = (D_i,M_i)$, for $i\in \Sbar$, where $D_i$ is a uniformly random value in $[k]$, and
$$
	M_i = X_i^{\mt} \quad \mbox{ if } D_i = t 	
$$
In other words, $D_i$ specifies which player's question to omit; the other $(k-1)$ players are fixed.

For $i \notin S$, we let $\Omega_\mi$ denote $\Omega$ with $\Omega_i$ omitted. We let $R_\mi := (\Omega_\mi,A_S)$. $R_i$ will refer to $\Omega_i$. We will use lowercase letters to denote instantiations of these random variables: e.g., $r_\mi$, $x^t_i$ refer to specific values of $R_\mi$, $X^t_i$ respectively.

\begin{claim}
\label{clm:ind}
	Conditioned on $R$, $\{\Xvec^t\}_{t\in [k]}$ are independent.
\end{claim}

In the following, $\P_I$ denotes the distribution of a uniformly random $i \in [m]$, and ``$\P \approx_\delta \Q$'' indicates that the probability distributions $\P$ and $\Q$ are $\delta$-close in statistical distance. We will fix
$$
	\delta = \frac{1}{m} \left ( \log \frac{1}{\P(W_S)} + |S| \log |\A|\right).
$$
The next lemma states that for an average $i$, if we sample questions $x_i,\xhat_i$ from the joint probability distribution $\rho(x_i,\xhat_i)$, the distributions of the corresponding dependency-breaking variables will be close.
\begin{lemma}
\label{lem:r}
	\begin{align*}
	\frac{1}{m} \sum_i \, \sum_{x_i,\xhat_i \in \X} \rho(x_i,\xhat_i) \left \| \P_{R_\mi | x_i, W_S} - \P_{R_\mi | \xhat_i, W_S} \right \|_1 \leq O(\sqrt{\delta})
	\end{align*}
	where $\rho(\cdot,\cdot)$ is the weight function of the $(k-1)$-connection graph $H_G$.
\end{lemma}
\begin{proof}
First, we establish the following: for all $t \in [k]$, we have
	\begin{align}
	\Ex_i \, \sum_{x_i^\mt, x_i^t,\xhat_i^t} \P_{X_i^{\mt}}(x_i^{\mt}) \, \P_{X_i^t | x_i^\mt}(x_i^t) \cdot \P_{\what{X}_i^t | x_i^\mt}(\xhat_i^t) \,  \left \| \P_{R_\mi | x_i, W_S} - \P_{R_\mi | \xhat_i, W_S} \right \|_1 \leq O(\sqrt{\delta})
	\label{eq:r1}
	\end{align}
	where we use the shorthand $x_i := x_i^\mt x_i^t$ and $\xhat_i := x_i^\mt \xhat_i^t$. This follows from the same arguments found in~\cite{Holenstein2009,BravermanG2015}; for each player $t$, we can treat the other $(k-1)$ players as one ``meta player'', and apply the two-player analysis to obtain~\eqref{eq:r1}.
	
	Observe that when $x_i^t \neq \xhat_i^t$, we have $$ \P_{X_i^{\mt}}(x_i^{\mt}) \cdot \P_{X_i^t | x_i^\mt}(x_i^t) \cdot \P_{\what{X}_i^t | x_i^\mt}(\xhat_i^t) = k \rho(x_i,\xhat_i).$$ On the other hand, when $x_i^t = \xhat_i^t$, $x_i = \xhat_i$ so therefore $\left \| \P_{R_\mi | x_i, W_S} - \P_{R_\mi | \xhat_i, W_S} \right \|_1 = 0$. Furthermore, for $x_i$ and $\xhat_i$ that differ in more than $1$ coordinate, we have $\rho(x_i,\xhat_i) = 0$, and for every $x_i, \xhat_i$ such that $\rho(x_i, \xhat_i) \neq 0$, there exists a unique $t\in [k]$ such that $x^t_i \neq \xhat_i^t$. Thus we can bound for every $i$:
	\begin{align*}
		&\sum_{x_i,\xhat_i \in \X} \rho(x_i,\xhat_i) \left \| \P_{R_\mi | x_i, W_S} - \P_{R_\mi | \xhat_i, W_S} \right \|_1 \\
		&= \frac{1}{k} \sum_{t\in [k]}  \sum_{x_i^\mt,x_i^t,\xhat_i^t} \P_{X_i^{\mt}}(x_i^{\mt}) \cdot \P_{X_i^t | x_i^\mt}(x_i^t) \cdot \P_{\what{X}_i^t | x_i^\mt}(\xhat_i^t) \,  \left \| \P_{R_\mi | x_i, W_S} - \P_{R_\mi | \xhat_i, W_S} \right \|_1.
	\end{align*}
	Averaging over $i$ and using~\eqref{eq:r1}, we obtain the statement of the lemma.
\end{proof}

\subsection{Avoiding correlated sampling using expansion}\label{sec:avoid-correlated}
At this point, ideally, every player would like to sample from $R_\mi | x_i, W_S$. Lemma~\ref{lem:r} establishes that $R_\mi | x_i, W_S$ is close to $R_\mi | x^{\mt}_i, W_S$ for each $t\in [k]$. None of the players alone has knowledge of $x^{\mt}_i$, however. We will show now that nevertheless, there is a \emph{global} distribution known to all the players, from which the players can approximately sample $R_\mi | x_i, W_S$.

\begin{lemma}\label{lem:global}
	For all $i \in [m]$ there exists a distribution $\wt{\P}_{R_\mi}$ over $R_\mi$ such that
$$
	\frac{1}{m} \sum_i \sum_{x} \rho(x) \| \P_{R_\mi | x, W_S} - \wt{\P}_{R_\mi} \|_1 \leq O \left ( \frac{\delta^{1/4}}{ \sqrt{\lambda}} \right ).
$$
\end{lemma}
\begin{proof}
For each $i$, define the vector-valued function $g_i: \X \to \R^{R_\mi}$ as follows: for all $x \in \X$,\footnote{Here, when we write $x$, we are implicitly mean $x_i$; we drop the subscript $i$ for notational convenience.}
$$
	g_i(x) = \sqrt{\P_{R_\mi | x,W_S}}
$$
where $\sqrt{\P_{R_\mi | x,W_S}}$ denotes the entry-wise square root of the probability distribution \\ $\P_{R_\mi | x,W_S}$, viewed as a vector. In other words, the entries of $g_i(x)$ are indexed by different values $r_\mi$ of the random variable $R_\mi$. Thus, $g_i$ is a unit vector in the $\ell_2$ norm.

For any $i$ and any $x,\xhat \in \X$, the quantity $\| g_i(x) - g_i(\xhat) \|^2$ is simply the square of the \emph{Hellinger distance} between $\P_{R_\mi | x,W_S}$ and $\P_{R_\mi | \xhat,W_S}$, which can be related to their statistical distance by
$$
\| g_i(x) - g_i(\xhat) \|^2 \leq \| \P_{R_\mi | x,W_S} - \P_{R_\mi | \xhat,W_S} \|_1.
$$
By Lemma~\ref{lem:r}, we can average the above inequality over all $i$ and choosing
$x,\xhat$ according to the probability distribution $\rho(x,x')$, we get
\begin{align*}
	\frac{1}{m} \sum_i \sum_{x,\xhat} \rho(x,\xhat) \| g_i(x) - g_i(\xhat) \|^2   \leq \Ex_i \sum_{x,\xhat} \rho(x,\xhat) \| \P_{R_\mi | x,W_S} - \P_{R_\mi | \xhat,W_S} \|_1 \leq O(\sqrt{\delta})
\end{align*}
But now we can leverage Equation~\eqref{eq:expander}. For every $i$, define the vector $\overline{g}_i = \sum_{x} \P_{X}(x) g_i(x)$. This is not necessarily a unit vector, but we have the relation
\begin{align*}
	\frac{1}{m} \sum_i \sum_{x} \rho(x) \| g_i(x) - \overline{g}_i \|^2 &\leq \frac{1}{\lambda m} \sum_i \sum_{x,\xhat} \rho(x,\xhat) \, \left \| g_i(x) - g_i(\xhat) \right \|^2  \leq O \left ( \frac{\sqrt{\delta}}{\lambda} \right).
\end{align*}
If $O(\sqrt{\delta}/\lambda )$ is small, then this implies that on average, the vectors $g_i(x)$ are all close to a fixed state $\overline{g}_i$. Since $g_i(x)$ are all unit vectors, this implies that $\overline{g}_i$ is close to a unit vector. By increasing the error by a constant factor, we can assume that $\overline{g}_i$ is in fact a unit vector. Thus we can construct the probability distribution
$$
	\wt{\P}_{R_\mi}(r_\mi) = \overline{g}_i(r_\mi)^2.
$$
Using that the statistical distance is at most (up to constant factors) the square root of the Hellinger distance, we get that
$$
	\frac{1}{m} \sum_i \sum_{x} \rho(x) \| \P_{R_\mi | x, W_S} - \wt{\P}_{R_\mi} \|_1 \leq O(\delta^{1/4} \lambda^{-1/2}).
$$
\end{proof}

\subsection{Finishing the proof}
Let $\{ f^t \}$ be an optimal strategy for the game $G^{\otimes n}$. If $\P(W) \leq \frac{4}{\eps} 2^{- \eps n/16}$, then we are done. Otherwise, suppose $\log 1/\P(W) \leq \eps n/16 - \log 4/\eps$. Let the subset $S$ be as given by Proposition~\ref{prop:subset}, and assume the coordinates are numbered so that $S$ is the last $|S|$ coordinates of $[n]$. For all $i \in [m]$, let $\wt{\P}_{R_\mi}$ be as given by Lemma~\ref{lem:global}. Consider the following single-shot strategy by the players, where $x$ is drawn from $\mu$ and $x^t$ is given to player $t$:
\begin{enumerate}
	\item Using shared randomness, the players sample an $i \in [m]$ uniformly at random, and sample $r_\mi$ from $\wt{\P}_{R_\mi}$. Each player $t$ then sets $x_i^t$ to be their ``true'' question $x^t$ they received from the referee.
	\item Using private randomness, each player $t$ samples $x^t_{-i}$ from $\P_{X^t_{-i} | x_i^t, r_\mi}$. That is, each player samples questions for the $n$ coordinates that come from the repeated game, conditioned on their own true input $x_i^t$ and the dependency-breaking variable $r_\mi$.
	\item Player $t$ outputs the $i$'th component of the answer vector $f^t(x^t_{[n]})$.
\end{enumerate}
Lemma~\ref{lem:global} implies that after the first step, the sample $r_\mi$ each player possesses will be, up to statistical error $O(\delta^{1/4}/\sqrt{\lambda })$, distributed according to $\P_{R_\mi | x, W_S}$ (on average over $i$ and $x$). Then, by Claim~\ref{clm:ind}, the joint distribution of the random variables $\{ X_{[n]}^t \}$ that the players have sampled is $$
	\P_{X_{[n]}^1 | x_i^t, r_\mi} \times \cdots \times \P_{X_{[n]}^k | x_i^t, r_\mi} = \P_{X_{[n]} | x_i, r_\mi}.
$$

Thus, conditioned on $r_\mi$ and $x_i$, the distribution of their answers $a_i$ will be distributed according to $\P_{A_i | x_i, r_\mi}$. Averaging over $i$, $x_i$, and $r_\mi$, we get that their answers are $O(\delta^{1/4}/\sqrt{\lambda})$-close to being distributed according to

$$
	\P_I \cdot \P_{X_i} \cdot \P_{R_\mi | X_i, W_S} \cdot \P_{A_i | X_i, R_\mi}
$$
where $\P_I$ stands for the uniform distribution over $i \in [m]$. We also have that, on average $i$, $\P_{X_i | W_S}$ is $O(\sqrt{\delta})$-close in statistical distance to $\P_{X_i}$. Thus their answers are $O(\delta^{1/4}/\sqrt{\lambda}) + O(\sqrt{\delta})$ close to being distributed as
$$
	\P_I \cdot \P_{A_i | W_S}.
$$
Thus by Proposition~\ref{prop:subset}, the probability that the players win $G$ is at least
$$
	1 - \eps/2 - \left ( O(\delta^{1/4}/\sqrt{\lambda}) + O(\sqrt{\delta}) \right).
$$
If $O(\delta^{1/4}/\sqrt{\lambda}) + O(\sqrt{\delta}) < \eps/2$, then we would contradict the fact that $\val(G) = 1 - \eps$. This implies that we must have $\delta = \Omega( \eps^4 \lambda^2 )$. If we let $\P(W) = 2^{-\gamma n}$, then we can write
$$
	\delta \leq \frac{16}{\eps} \left [ \frac{1}{n} \log \frac{4}{\eps} + 2 \log |\A| \gamma \right ]
$$
where we plugged in the bound on $|S| \leq n/2$ from Proposition~\ref{prop:subset}. This implies the lower bound
\begin{align}
	\gamma \geq \Omega\left ( \frac{\eps^5 \lambda^2}{\log |\A|} \right)
\end{align}
when $n \geq \frac{\log 4/\eps}{\eps^5 \lambda^2}$, proving the theorem.

%This implies that the players, upon input $x\in \X$, can sample approximately from $\P_{R_\mi | x,W_S}$, by sampling from $\wt{\P}_{R_\mi}$. Then the players could produce answers approximately distributed according to $\P_{A^1_i, \ldots , A^t_i | x,r_\mi}$, which would allow them, using Proposition~\ref{prop:subset}, to win $G$ with probability at least $1-\eps/2 -\delta^{1/4}/ \sqrt{\lambda}$, if:

%$$
%	\P(W_S) \geq \exp \left ( -\frac{\delta n}{\log |\A|} \right ).
%$$
%
%This would lead to a contradiction for a small enough $\delta \approx \eps ^4 \lambda^2/100$.

{\small
\bibliographystyle{plainurl}
\bibliography{mplayer}
%\bibliography{/Users/prahladh/Documents/LaTeX/papers/jrnl-names-abb,/Users/prahladh/Documents/LaTeX/papers/prahladhbib,/Users/prahladh/Documents/LaTeX/papers/crossref}
}

\appendix

\section{Proof of Corollary~\ref{cor}}\label{sec:cor}

For each type of game, we compute a lower bound on the second-smallest eigenvalue of the corresponding $(k-1)$-connection graph. Applying Theorem~\ref{thm:expanding-main} then yields the statements of the corollary.

\subsection{Free games}

For simplicity, assume that $\mu(x)$ is the uniform distribution over $[d]^k$, where $d = |\X^1| = \cdots = |\X^k|$.\footnote{Indeed, by letting $d$ be large enough, we can approximate $\mu$ arbitrarily well through discretization and identifying $[d]$ with $\X^t$ for $t= 1,\ldots,k$ in a many-to-one-fashion. Our bounds will not depend on $d$, so $d$ can be taken to be arbitrarily large.} Then the $(k-1)$-connection graph is a weighted version of the $d$-ary, $k$-dimensional hypercube (with self loops). Indeed, the corresponding weight function $\rho$ behaves as follows: for $x,x' \in [d]^k$, we have $\rho(x,x) = d^{-(k+1)}$, and $\rho(x,x') = d^{-(k+1)}/k$ when $x$ and $x'$ differ in exactly one coordinate, and is $0$ otherwise. If we compute the normalized Laplacian $L_H$, we get that
\begin{align*}
	(L_H)_{u,v} = \left \{ \begin{array}{ll}
		1 - \frac{1}{d}	& \mbox{ if } u = v \mbox { and } \rho(v) \neq 0 \\
		- \frac{1}{kd} & \mbox { if } \rho(u),\rho(v) \neq 0 \\
		0 & \mbox{ otherwise }
	\end{array}
	\right.
\end{align*}
This is the normalized Laplacian corresponding to the Cayley graph over the Abelian group $(\Z/d\Z)^k$ with (weighted) generators $\{ g \in (\Z/d\Z)^k : |g| \leq 1 \}$ where $|g|$ is the number of non-zero components of $g$. If $g = (0,0,\ldots,0)$, then its weight is $d^{-1}$, and if $|g| = 1$, then its weight is $(kd)^{-1}$. The spectrum of Cayley graphs is well understood; we have that the smallest non-zero eigenvalue of $L_H$ is therefore $\lambda(H) = \frac{1}{k}$. Thus $H_G$ is a $1/k$-expander.

\subsection{Anchored games}

%In \cref{sec:anchor-eigenvalue}, we prove that for $\alpha$-anchored games (see \cref{def:anchored-game}), $\lambda(H_G) \geq  8k/\alpha^{k}$. %Plugging in this bound gives us
%$$ \val(G^{\otimes n}) \leq \exp \left(-\frac{c \, \alpha^{2k} \,\eps^5 \,  n}{64 ~k^2\log |\A|}\right) $$

Here, we prove a lower bound on the second eigenvalue of the $(k-1)$-connection graph of an \emph{anchored} game, and show that it is at least $8k/\alpha^{k}$. Plugging in this bound into Theorem~\ref{thm:expanding-main} gives us
$$ \val(G^{\otimes n}) \leq \exp \left(-\frac{c \, \alpha^{2k} \,\eps^5 \,  n}{64 ~k^2\log |\A|}\right). $$
This asymptotically matches the bounds obtained in ~\cite{BavarianVY2015} in terms of the dependence on $\alpha$ and $k$.

Let us first recall the definition of an anchored game.

\begin{definition}[$\alpha$-anchored games~\cite{BavarianVY2015}] \label{def:anchored-game}
Given a $k$-prover game $G$, and a parameter $\alpha < 1$ we define the \emph{$\alpha$ anchored} game $G_{\bot}$ as follows: the referee chooses a question tuple $(x^1, \ldots, x^k)$, according to $G$, and independently, for every $t\in [k]$, replaces $x^t$ by the anchoring symbol $\bot$ with probability $\alpha$ to get the tuple $(x'^1, \ldots x'^k)$. The new domain is thus $\X'^1 \times \X'^2 \ldots \X'^k$, where $X'^i= \X \cup \{\bot\}$. If any of the $x'$'s are $\bot$, the verifier accepts trivially, otherwise the verifier accepts according to the predicate of the game $G$.
\end{definition}

For convenience, we will denote the $\alpha$-anchored game itself by $G$ in this section, and its $(k-1)$-connection graph by $H_G$. We will show the following lemma.

\begin{lemma}\label{lem:eig-anchor}\footnote{Although the proof of the lemma can be easily seen to show a bound dependent only on $\alpha$ and $k$ for all $\alpha<1$, the anchored game definition in ~\cite{BavarianVY2015} sets $\alpha$ to be a constant $<1/2$. We only state this case, for clarity of exposition and comparison to their result.}~~
$\lambda(H_G)\geq \alpha^{k}/8k$, when $\alpha < 1/2$.
\end{lemma}

In order to prove Lemma~\ref{lem:eig-anchor}, we need to make a couple of observations. First, note that the $1$-connection graph $H_G$'s vertices can be partitioned into disjoint sets $V_0, V_1, \ldots, V_k$, where $V_i$ has vertices of all question-tuples with exactly $i$ bottom symbols. Thus, $V_0$ has vertices corresponding to the original question tuples, and $V_k = \{(\bot, \bot, \ldots, \bot )\}$. While $V_0$ has edges between its own vertices (corresponding to edges in the $1$-connection graph of the un-anchored game), all other edges in $H_G$ go between $V_i$ and $V_{i+1}$.

We will lower bound $\lambda(H_G)$ using the notion of \emph{congestion} in the graph. This technique was first introduced by Diaconis and Strook~\cite{DiaconisS1991}, and improved by Sinclair~\cite{Sinclair1992}. The below form can be found in the survey~\cite[Section 4]{Guruswami2016}.

\smallskip
Let us view $H_G$ as an undirected graph\footnote{On the other hand, if viewed as a directed Markov chain, the transition probability $\Pr[ y \mid x] $ for moving from $x$ to $y$  is exactly $\rho(x,y)/\rho(x)$. The stationary distribution on every vertex is $\rho(x)$.}, with weight function $\rho$ on the edges. Since $\rho(x,y) = \rho(y,x)$ by our definition, this is well-defined. A set of canonical paths in $H_G$  is a set $\mathcal{P}$ of simple paths, one  between every ordered pair $(x,y)$ in $H_G$. The \emph{path congestion parameter} of this set of canonical paths is defined as follows:

\[
\zeta(\mP) \; \triangleq \; \max_{e \in E(H_G)}\; \frac{1}{\rho(e)} \sum_{p_{xy}\ni e } \rho(x)\rho(y)|p_{xy}|
\]

Above, $p_{xy}$ denotes the path from $x$ to $y$ in $\mathcal{P}$, and $|p_{xy}|$ is its length. Intuitively, the numerator in the above equation defines the `load' on the edge $(x,y)$, while $\rho(x,y)$ can be interpreted as its capacity. Thus, one would naturally expect that if we could find a set of canonical paths with low congestion parameter, the graph must be expanding in some sense. This is formalized in the following theorem:

\begin{theorem}[~\cite{Sinclair1992}, see also {~\cite[Theorem 4.3]{Guruswami2016}}]\label{thm:cong-lb}
For any set of canonical paths $\mP$,
\[\lambda(H_G) \geq \frac{1}{\zeta(\mP)}\]
\end{theorem}

We will prove Lemma~\ref{lem:eig-anchor}, by choosing a good set of canonical paths in $H_G$.

\begin{proof}[Proof of Lemma~\ref{lem:eig-anchor}]
Consider two vertices $x,y$ in $H_G$. Let $\Delta(x,y) = \{i_1, \ldots, i_s \} \subseteq [k]$ be the set of (player) indices  where the tuples differ, with $i_1 \leq i_2 \leq \ldots i_s$. We will define the canonical path from $x$ to $y$ to be the one obtained by flipping each of $x^{i_1}, \ldots x^{i_s}$ to $\bot$ in order, and then flip these to $y^{i_1}, \ldots y^{i_s}$, but in the reverse order $i_s\rightarrow \ldots  \rightarrow i_1$. Each flip corresponds to moving along an edge in $H_G$. Call the set of these canonical paths $\mP$. The path from $x$ to $y$ in $\mP$ is exactly the reverse of the path from $y$ to $x$.

We will upper bound the congestion through any edge $e = \{u,v\}$ caused by $\mP$. If $u,v \in V_0$, then no path in $\mP$ passes through this edge, and hence the congestion  on $e$ is $0$. Suppose that $u \in V_l$, and $v \in V_{l+1}$ for some $l <k$.

We need to find which vertices $x$ would use a canonical path that passes from $u$ to $v$ to reach another vertex. To identify this set, define $B_v \triangleq \{i\in[k] \;:\; v_i = \bot\}$, and similarly $B_u \triangleq \{i\in[k] \;:\; u_i = \bot\}$. Clearly $|B_v| = l+1, |B_u|=l$, and $B_u \subseteq B_v$. Let us write $u$ as $u=(\bot^l, z_u)$, where the indices are appropriately ordered (with $z_u$ in $\overline{B}_u)$.

For $0\leq r \leq l$, a vertex $w \in V_r$ will be said to be in the $r$-th \emph{shadow} of $u$ (denoted by $S_r(u)$), if: \begin{itemize} \item[(a)] $w|_{\,\overline{B}_u} = z_u$, and \item[(b)] If $B_u= \{j_1, \ldots, j_l$\}, with $j_1 \leq \ldots \leq j_l$, then $w_{j_q} \neq \bot$ for every $q > l-r$ . \end{itemize}

The following Claim is easy to verify:
\smallskip
\begin{claim}
$\rho(S_r(u)) = \Pr_{x\sim \rho}[x^{\,\overline{B}_u}=z_u] \times \,\alpha^r(1-\alpha)^{l-r}$
\end{claim}

\begin{proof}
Any vertex in $S_r(u)$ can be seen to be generated by the verifier in the following way: pick a random question in the original (un-anchored) game conditioned on $x^{\,\overline{B}_u}=z_u$, then flip $j_1, \ldots j_r$ to $\bot$ (happens with probability $\alpha^r$), and leave the others unflipped (happens with probability $(1-\alpha)^{k-r}$). The probability of not flipping $\overline{B}_u$  (i.e. $(1-\alpha)^{k-l}$) is accounted for in the distribution $\rho$ of the anchored game.  This yields the measure of the set $S_r(u)$ as being the expression given above.
\end{proof}

Any path in $\mP$ that passes through $u$ will necessarily either originate or end in one of its shadows. The length of any canonical path as defined above is at most $2k$. Hence, the load through the edge $(u,v)$ can be upper bounded as follows (denoting $\Pr_{x\sim \rho}[x^{\,\overline{B}_u}=z_u]\,$ by $\,\Pr[z_u]$ for clarity):

\begin{align*}
\sum_{p_{xy}\ni e } \rho(x)\rho(y)|p_{xy}| &\leq 2k \sum_{p_{xy}\ni e } \rho(x)\rho(y)\\
& \leq \, 4k \sum_{r=0}^{l}~ \rho(S_r(x))\\
%& \qquad \qquad \qquad \ldots \text{since paths $p_{xy}$ and $p_{yx}$ are reverse of each other}\\
& = \, 4k \sum_{r=0}^{l}~ \Pr_{x\sim \rho}[x^{\,B_u}=z_u] \times\,\alpha^r(1-\alpha)^{l-r}\\
& =  \, 4k(1-\alpha)^l~\Pr[z_u] \sum_{r=0}^l ~ \left( \frac{\alpha}{1-\alpha}  \right)^r\\
&\leq \, 4k (l+1)(1-\alpha)^l ~ \Pr[z_u] &\ldots \text{ since } \alpha <1/2\\
&\leq\, 8k ~\Pr[z_u]
\end{align*}

The capacity of edge $(u,v)$ is $\rho(u,v)=\Pr[z_u] \times \alpha^{l}$. Thus, the congestion along the edge is bounded by

$$
\zeta(e) \leq \frac{8k ~\Pr[z_u]}{\Pr[z_u] \times \alpha^{l}} = \frac{8k}{\alpha^{l}}
$$

Hence, the maximum congestion is bounded by $\zeta(\mP) \leq \frac{8k}{\alpha^{k}} $, which yields the lower bound $\lambda(H_G) > \frac{\alpha^{k}}{8k}$, by invoking Theorem~\ref{thm:cong-lb}.
\end{proof}

%This asymptotically matches the bounds obtained in ~\cite{BVY15anchoring} in terms of the dependence on $\alpha$ and $k$.
%
%\medskip
%\noindent (\emph{Fortified games.})

\subsection{Connected games}
This follows from the observation that $\lambda(H) \geq \rho_{min}$ when the graph $H$ is connected. The ``in particular'' statement follows from the definition of the weight function $\rho$ of the $(k-1)$-connection graph: $\P_X(x)$ is simply $1/|S|$, and $\P_{X^t | x^\mt}(x'^t)$ is also at least $1/|S|$.

\end{document}